\pdfoutput=1 
\documentclass[12pt]{iopart}
\usepackage{setstack}
\usepackage{graphicx,amssymb,amsthm,epsfig,amsbsy,amsfonts,mathrsfs}
\usepackage[english,francais]{babel}

\usepackage{color}

\newtheorem{theorem}{Theorem}
\newtheorem{corollary}{Corollary}[section]

\theoremstyle{definition}

\theoremstyle{remark}

\def\d{\mathrm{d}}
\def\e{\mathrm{e}}
\def\i{\mathrm{i}}

\renewcommand{\leq}{\leqslant}
\renewcommand{\geq}{\geqslant}

\renewcommand\digamma{\Psi}

\newcommand{\Nc}{m}

\begin{document}

\renewcommand{\labelitemi}{$\bullet$}
\renewcommand{\labelitemii}{$\star$}

\selectlanguage{english}

\title[1D Anderson localisation and products of random matrices]
  {Lyapunov exponents, one-dimensional Anderson localisation and products of random matrices}  

\author{Alain Comtet}
\address{Univ. Paris-Sud, LPTMS, UMR 8626 du CNRS, 91405 Orsay, France \\
UPMC Univ. Paris 6, 75005 Paris, France}

\author{Christophe Texier}
\address{Univ. Paris-Sud, LPTMS, UMR 8626 du CNRS, 91405 Orsay, France\\
 LPS, UMR 8502 du CNRS, 91405 Orsay, France}   

\author{Yves Tourigny}
\address{School of Mathematics\\
        University of Bristol\\
        Bristol BS8 1TW, United Kingdom}



\date{\today}


\begin{abstract}
The concept of Lyapunov exponent has long occupied a central place in
the theory of Anderson localisation; its interest in this particular
context is that it provides a reasonable measure of the localisation
length. The Lyapunov exponent also features prominently in the theory
of products of random matrices pioneered by Furstenberg.
After a brief historical survey, we describe some recent work that
exploits the close connections between these topics. We review
the known solvable cases  of disordered quantum mechanics 
involving random point scatterers
and discuss a new solvable case.
Finally,
we point out some limitations of the Lyapunov exponent
as a means of studying localisation properties.
\end{abstract} 

\ams{ 60B20 , 60G51 , 82B44 }


\pacs{72.15.Rn , 02.50.-r}




\maketitle


\section{Introduction}

{\em Anderson localisation} is the term used to describe a generic phenomenon, discovered in the late fifties by P. W. Anderson, whereby the addition of a certain amount
of disorder or randomness in an otherwise deterministic medium causes
the waves propagating in the medium to become localised in space \cite{An}.
For quantum systems, the
understanding of transport properties requires a
thorough study of how the presence of disorder affects the nature of
quantum states--- a question first addressed in an earlier paper of
Landauer and Helland \cite{LAN}. Since that time, the one-dimensional
case has been discussed extensively and has led to a better
understanding of the physical mechanisms that are responsible for
localisation. A remarkable feature of one-dimensional
systems is that almost all states become localised as soon as there is {\em any} disorder. 
This result, first conjectured by Mott and Twose \cite {MT}, was
made more rigorous by Borland \cite{BORL} who considered an infinite
chain of identical localised potentials separated by regions of zero
potential. Assuming that the lengths of these regions are independent random
variables with the same probability distribution, Borland studied the growth rate of the wave function on a
semi-infinite chain with prescribed boundary conditions at one end. He then
argued, by using a ``matching argument'', that the positivity of the growth rate implies
the exponential localisation of the
wave functions. A purely mathematical proof using the properties of
transfer matrices was given by Matsuda and Ishii \cite{MI}. 

These physical arguments can in fact be made completely
rigorous. 
The exponential growth of the solutions of the Cauchy problem is a crucial
feature of the proof that the spectrum has no absolutely continuous component \cite{PAS},
and that it is pure-point \cite{GMP}. 
If $\psi(x,E)$ denotes a solution of the Cauchy problem (i.e. a solution
of the Schr\"odinger equation on the
positive half-line subject to boundary conditions at $x=0$)
then  the quantity
\begin{equation}
  \gamma(E) := \lim_{x \rightarrow \infty} \frac{\ln |\psi(x,E) |}{x}
  \label{gamma}
\end{equation}
is a self-averaging quantity
called the
{\it Lyapunov exponent} of the disordered system. 
A rigorous demonstration of the fact that, under certain hypotheses, $\gamma$ also quantifies the exponential decay of the {\em eigenfunctions}---
and therefore that its reciprocal can serve as a definition of the localisation length---
appears for
the first time in Ref.~\cite{CAR}.
It should be borne in mind, however, that this definition of the localisation length is only useful if certain conditions are fulfilled;
cases where the definition is inappropriate will be considered briefly
in \S \ref{conclusionSection}.

In order to illustrate the localisation phenomenon, Anderson made
use of a model in which the wave function solves a difference equation.
The general solution of this ``tight-binding'' model takes the form of
a product of random matrices, say,
\begin{equation}
  \Pi_n := M_n M_{n-1} \cdots M_1
  \label{product}
\end{equation}
where the $M_j$ are independent and identically-distributed square
matrices with a common probability measure $\mu(\d M)$. The quantity
\begin{equation}
  \gamma_{\mu} 
  := \lim_{n \rightarrow \infty} \frac{{\mathbb E} \left (\ln  |\Pi_n | \right )}{n}
  \label{gammaMu}
\end{equation}
is called the Lyapunov exponent of the product of random matrices and,
as we shall soon see, is effectively the same as $\gamma$.

Anderson localisation has been a powerful motivation for the
study of products of random matrices.
For this reason, the search for precise conditions 
on the measure $\mu(\d M)$ that
would guarantee the existence and the positivity of the Lyapunov
exponent $\gamma_{\mu}$ has been of particular interest. The main result in this respect is due to
Furstenberg \cite{Fu}--- a result
published the same year as Borland's paper; 
see also Oseledec \cite{Os} whose work can be considered as an extension  to dynamical systems of the work of Furstenberg and  Kesten \cite{FK} .
Further developments of these results, and their application to
the study of localisation,
are described in the works of Ishii \cite{Is}, Bougerol and
Lacroix~\cite{BL}, Carmona and Lacroix \cite{CaLa},  
Lifshits {\it et al.} \cite{LGP},  Luck \cite{Luc92},
Pastur and Figotin \cite{PF}, Crisanti
{\it et al.} \cite{CriPalVul93} and the references therein.

An important milestone in the development of the theory of localisation was the discovery of
the relationship between the Lyapunov exponent and the integrated
density of states $ N(E) $. It turns out that the characteristic
function 
\begin{equation}
\Omega(E) := \gamma(E) -\i\pi N(E)\,,
\label{characteristicFunction}
\end{equation}
viewed  as  a complex-valued
function on $\mathbb{R}$, is {\em analytic} in
the upper half of the complex plane \cite{HJ,T}. Interestingly,
this analyticity property of the characteristic function  was
in fact exploited much earlier by Dyson in his famous paper on the
dynamics of a disordered chain~\cite{Dy}.  

 Products of random matrices appear naturally in several other
 problems related to the physics of disordered systems---
see the monographs \cite{CriPalVul93,Luc92}---  and is still the subject of active research \cite{For}.
The  concept
 of  Lyapunov exponent  is a very useful tool for analysing a large
 class of  systems with quenched disorder, such as magnetic systems. A well-known 
 prototype is the Ising model in a random magnetic field, where  the
 calculation of the free enegy reduces to analysing an
 infinite product of $2\times 2$ matrices; in this context, the
 Lyapunov exponent is proportional to the free energy per spin.

The group from which the matrices $M_n$ in the product~(\ref{product})
are drawn
varies not only with the physical context but also with the choice of
vector basis. In the analysis of the Schr\"odinger equation in a random
potential, a standard choice is to consider the vector formed by the
wave function and its derivative $(\psi'(x)\,,\,\psi(x))$. 
As we shall see later on, the evolution of this vector is governed
by matrices belonging to the group $\text{SL}(2,{\mathbb R})$. This is the
formulation chosen for example in Refs.~\cite{CTT1,ComLucTexTou12} and
in the present article. 
If, instead, one chooses to focus on the scattering aspects of the problem
then the wave function, in a region where $V(x)=0$, is a combination 
of incoming and outgoing waves
$\psi(x)=A\,\e^{\i kx}+B\,\e^{-\i kx}$. 
In this setting, it is natural to consider transfer matrices $T$ connecting
pairs of complex amplitudes $(A,B)$. Current conservation then implies
that such transfer matrices belong to the group $\text{U}(1,1)$; see
the appendix of Ref.~\cite{CTT1} for further details,
and Ref.~\cite{Pen94} for a
pedagogical presentation and a review of the literature. 

Interestingly, this scattering formulation provides a natural way of 
generalising the one-dimensional model--- which is associated with
products of 
$2\times2$ matrices---
to a quasi-one-dimensional model where the matrices in the product are 
$2\Nc\times2\Nc$, where $\Nc$ is the
number of conducting channels.
The localisation problem then involves--- not just one--- but rather $\Nc$ (counting multiplicity)
Lyapunov exponents $\gamma_1\leq\cdots\leq\gamma_{\Nc}$.  
Since the pioneering work of Dorokhov \cite{Dor82}, whose results were later
rediscovered independently by Mello, Pereyra and
Kumar~\cite{MelPerKum88}, this 
topic has attracted a lot of attention owing to its relevance in the
description of weakly disordered metallic wires; see the review~\cite{Bee97}. 
These early works relied on the so-called {\em isotropy assumption}, namely that
each elementary slice of disordered metal redistributes the current uniformly 
amongst the $\Nc$ conducting channels. This assumption
produces a set of
Lyapunov exponents with the behaviour $\gamma_j=\gamma_1\,[1+\beta\,(j-1)]$, where
$\beta\in\{1,\,2,\,4\}$ is the Dyson index~\cite{Bee97}.
The smallest Lyapunov exponent, which scales  with the number of channels like
$\gamma_1\propto1/\Nc$, is usually
interpreted as the reciprocal of the localisation length. 
The ideas of Dorokhov and Mello {\em et al.} have since been extended to other symmetry classes of
disordered models; see the review \cite{EveMir08}.
Another line of research stimulated by these ideas
is to look for models where the isotropy hypothesis may be relaxed
\cite{ChaBer93,MelTom91,MutKla99,MutGop02}, with the aim of
studying the passage from one-dimensional to higher-dimensional
localisation.
To close these brief remarks on the multichannel case, we
point out that the Lyapunov exponent is at the heart of numerical
studies of localisation that rely on the scaling approach~\cite{KraMac93}.

The present paper revisits the interplay between products of random
matrices of $\text{SL}(2,{\mathbb R})$ and
one-dimensional Anderson localisation. Whereas most of the
works cited above emphasise {\it discrete} models obtained by making
the tight-binding approximation~\cite{CriPalVul93,Luc92}, we shall focus here
on one-dimensional continuous models  that make use of the notion of
{\it point scatterer}. One familiar example 
is the Kronig--Penney model \cite{KP} based on
delta-scatterers. We shall see that, by considering 
a natural generalisation of the concept of point scatterer, we
arrive at a useful interpretation of a general product 
of random matrices as a model of disorder. We illustrate the
fruitfulness of this interpretation 
by exhibiting concrete instances of the probability measure
$\mu(\d M)$ of the matrices $M_j$ for which the Lyapunov exponent may
be expressed in terms of special functions. 
One of these models is new and
has interesting connections with Sinai's study of diffusion in a
random environment~\cite{BCGL,Si}. 
Another justification for focusing on models involving generalised
random point scatterers is that--- far from being special--- they can on the contrary
be mapped onto the {\em most general product} of random matrices
in $\text{SL}(2,\mathbb{R})$. Hence such models
can in principle describe the whole of one-dimensional disordered quantum
mechanics.

The remainder of the paper is as follows: \S \ref{scattererSection}
describes the concept of generalised point scatterer, explains how
models of disorder may be constructed from them, and how completely general
products of matrices may be associated with such models.
We then list some known solvable cases.
\S \ref{susySection} is devoted to the analysis of a {\em new} example
where the Lyapunov exponent may be expressed in terms of the
generalised hypergeometric function. The paper ends in \S
\ref{conclusionSection} with a discussion of the limitations of the
Lyapunov exponent as a measure of localisation in disordered systems. 


\section{Generalised point scatterers and products of random matrices}
\label{scattererSection}

\subsection{Point scatterers and transfer matrices}

A  {\it point scatterer} is the idealised limit of a potential whose action is highly localised. The most familiar example is the delta-scatterer at a
point, say $x_j$. In the context of the Schr\"{o}dinger equation
 \begin{equation}
 - \psi''(x) + V(x)\, \psi(x) = E \,\psi(x)
   \,, \quad x > 0\,,
 \label{schroedingerEquation}
 \end{equation}
 this potential vanishes for every $x \ne x_j$ and is defined at
 $x=x_j$ via the boundary condition 
 \begin{equation}
 \left(\begin{array}{c}
 \phantom{k\,}\psi'(x_j+) \\
 \psi(x_j+)
 \end{array}\right)
 = \left(\begin{array}{cc}
 1 & u_j \\
 0 & 1
 \end{array}\right)
 \left(\begin{array}{c}
 \phantom{k\,}\psi'(x_j-) \\
 \psi(x_j-)
 \end{array}\right)
 \label{deltaScatterer}
 \end{equation}
 where $u_j$ is the strength of the scatterer, i.e. 
$V(x)=\sum_ju_j\,\delta(x-x_j)$. 
This delta--scatterer is generalised by replacing the $2 \times 2$
matrix on the right-hand side by an arbitrary 
matrix, say $B_j$. Thus 
\begin{equation}
 \left(\begin{array}{c}
 \psi'(x_j+) \\
  \psi(x_j+)
 \end{array}\right)
=B_j
 \left(\begin{array}{c}
 \psi'(x_j-) \\
  \psi(x_j-)
 \end{array}\right)\,.
 \label{generalisedScatterer}
 \end{equation}
 Conservation of the probability current requires
$B_j\in\text{SL}(2,{\mathbb R})$.
 If, instead of a single scatterer, we consider a sequence of scatterers placed at the points
 \begin{equation*}
  x_1 < x_2 < \cdots 
\end{equation*}
 then, assuming $E = k^2 > 0$, the solution of the Schr\"{o}dinger equation may be expressed in the form
 \begin{equation}
 \left(\begin{array}{c}
 \psi'(x_n+) \\
 \psi(x_n+)
 \end{array}\right)
= \Pi_n
 \left(\begin{array}{c}
 \psi'(x_1-) \\
  \psi(x_1-)
 \end{array}\right)\,.
\end{equation}
 where $\Pi_n$ is the product (\ref{product}) with
 \begin{equation}
 M_j = 
 \left(\begin{array}{cc}
 \sqrt{k} & 0 \\
 0 & \frac{1}{\sqrt{k}}
 \end{array}\right)
  \left(\begin{array}{cc}
 \cos\theta_j & -                 \sin\theta_j \\
 \sin\theta_j & \phantom{-}\cos\theta_j
 \end{array}\right)  \left(\begin{array}{cc}
 \frac{1}{\sqrt{k}} & 0 \\
 0 & \sqrt{k}
 \end{array}\right) \, B_j \,
\end{equation}
and
\begin{equation*}
\theta_j = k\,\ell_j\,,\;\;
  \ell_j:=x_{j+1}-x_{j}\,.
\end{equation*}
Now, by applying the Gram--Schmidt algorithm to the columns,  every $2 \times 2$ matrix $M$ with unit determinant may be expressed in the form
\begin{equation}
M = \left(\begin{array}{cc}
\cos \theta & -                   \sin \theta \\
\sin \theta &  \phantom{-} \cos \theta
\end{array}\right)
\left(\begin{array}{cc}
\e^w & 0 \\
0 & \e^{-w}
\end{array}\right)
\left(\begin{array}{cc}
1 & u \\
0 & 1
\end{array}\right)
\label{iwasawaDecomposition}
\end{equation}
for some real parameters $\theta$, $w$ and $u$. This is the Iwasawa
decomposition of $\text{SL}(2,{\mathbb R})$ into 
compact, Abelian and nilpotent subgroups. Therefore, by working with point scatterers of the form
\begin{equation}
B_j := \left(\begin{array}{cc}
\e^{w_j} & 0 \\
0 & \e^{-w_j}
\end{array}\right)
\left(\begin{array}{cc}
1 & u_j \\
0 & 1
\end{array}\right)
\label{doubleImpurity}
\end{equation}
and taking $E = k^2=1$, we obtain a correspondence between
general products of matrices $\Pi_n$ in $\text{SL}(2,{\mathbb R})$ 
and the Schr\"{o}dinger  
equation. The parameter $k$ may be easily reintroduced by simple
dimensional analysis.

\subsection{Product of random matrices---The Riccati variable}
\label{RiccatiSubsection}

There are several options in defining a disordered quantum system with point
scatterers. 
The randomness may be in the strength of the scatterer, i.e. in the matrix
$B_j$, or in the position of the scatterer, i.e. in the
coordinate $x_j$ or in {\em both} the strength and the position; see \cite{LGP} where various models 
are reviewed.

{\it We shall henceforth confine our attention to the particular case
  where the spacing between consecutive scatterers, i.e. the angle
    $\theta_j=k\ell_j$, is exponentially  distributed:}
\begin{equation}
  {\mathbb P} \left (\ell_j \in S \right ) 
  = p \int_{S \cap {\mathbb R}_+} \e^{-p \ell}\,\d \ell
  \label{eq:DistributionEllj}
\end{equation}
where $1/p>0$ is the mean spacing.
This is the situation that arises when impurities are dropped
uniformly on $\mathbb{R}$ with a mean density $p$.

The equation satisfied by the Riccati variable
\begin{equation}
Z(x) := \frac{\psi'(x)}{\psi(x)}
\label{riccatiVariable}
\end{equation}
is
\begin{equation}
Z'(x) = -E -Z^2(x) \quad \text{for $x \notin \{ x_j \}$}
\end{equation}
and
\begin{equation}
Z(x_j+) = {\mathcal B}_j \left ( Z(x_j-) \right ) \quad \text{for $j \in {\mathbb N}$} 
\end{equation}
where ${\mathcal B}_j$ is the linear fractional transformation associated with the matrix $B_j$, i.e.
\begin{equation}
{\mathcal B}_j(Z) := \e^{2 w_j} \left ( Z+u_j\right )\,.
\end{equation}
Because the spacing between consecutive scatterers is exponentially
distributed, $Z$ is a Markov process, and it was shown by Frisch and
Lloyd \cite{FL} (see also \cite{CTT1}) that the stationary density
$f(Z)$ satisfies 
\begin{equation}
\hspace{-1.5cm}
\frac{\d}{\d Z} \left [ (Z^2+E) f(Z) \right ] +p\,\int_{\text{SL}(2,{\mathbb R )}} \mu_B (\d B) \left [ f \left ( {\mathcal B}^{-1} (Z)  \right ) \frac{\d {\mathcal B}^{-1}(Z)}{\d Z} - f(Z)  \right ] = 0
\label{frischLloydEquation}
\end{equation}
where $\mu_B$ is the probability measure of the random matrix $B$.

The relationship between the Riccati variable and the Lyapunov exponent is easy to establish--- at least heuristically. A completely rigorous
treatment would follow the lines of Kotani's work \cite{Ko}.
Given the particular form of the upper-triangular matrix $B_j$, we have
\begin{eqnarray}
\nonumber
\ln \left | \psi (x_n+) \right | = - w_n + \ln \left | \psi (x_n-) \right |  \\
\nonumber
= -w_n + \int_{x_{n-1}}^{x_n} Z(x)\,\d x + \ln \left | \psi (x_{n-1}+) \right |  \\
= \cdots = - \sum_{j=1}^n w_j +  \int_{x_{1}}^{x_n} Z(x)\,\d x + \ln \left | \psi (x_{1}-) \right | \,.
\end{eqnarray}
Dividing by $x_n$ and letting $n \rightarrow \infty$, we obtain
\begin{equation}
\lim_{n \rightarrow \infty} \frac{\ln \left | \psi (x_n) \right |}{x_n}  = -\lim_{n \rightarrow \infty} \frac{1}{x_n} \sum_{j=1}^n w_j + \lim_{n \rightarrow \infty} \frac{1}{x_n} \int_{x_1}^{x_n} Z(x)\,\d x\,.
\end{equation}
Since 
\begin{equation*}
x_n-x_1 =  \sum_{j=1}^{n-1} \ell_j
\end{equation*}
where the $\ell_j$ are independent and identically distributed with mean $1/p$, the law of large numbers implies that
\begin{equation}
x_n \sim  n/p \quad \text{almost surely, as $n \rightarrow \infty$}\,.
\label{RelationXN}
\end{equation}
Then, by the ergodic theorem,
\begin{equation}
\gamma = - p\,{\mathbb E} (w) + \int_{\mathbb R} Z\, f(Z) \,\d Z
\label{riccatiFormula}
\end{equation}
almost surely, where $f$ is the density of the stationary distribution of the Riccati variable. The upshot is that the Lyapunov exponent of the system may in principle
be computed by solving the Frisch--Lloyd equation (\ref{frischLloydEquation}) for the stationary density, and then performing the integral
in Formula (\ref{riccatiFormula}).

\subsection{Correspondence with the Furstenberg theory}
\label{furstenbergSubsection}

It is instructive to compare this method of calculation with that
based on Furstenberg's theory \cite{Fu,GM}. 
The concept of direction, through the projective space $P \left ( {\mathbb R}^2 \right )$, plays a prominent part in that theory.
This is due to the fact that the product of matrices grows if and only if the columns of the product tend to align along a common direction.
In ${\mathbb R}^2$, a direction may be parametrised by a single number; let us choose the reciprocal $z \in {\mathbb R} \cup \{\infty\}$ of the slope. 
The columns of the product $\Pi_n$ have directions that are random.  When a vector of random direction, say $z$, is multiplied by a random matrix 
\begin{equation}
M = \left(\begin{array}{cc} 
a & b \\
c & d
\end{array}\right)
\end{equation}
a new vector is produced, whose direction is the random variable  
\begin{equation}
{\mathcal M}(z) := \frac{a z+b}{cz + d}\,.
\end{equation}
We say that the distribution, say $\nu_{\mu}(\d z)$, of $z$ is {\it invariant} under the action of matrices drawn from $\mu$ if the distributions of the old and the new directions are identical, i.e.
\begin{equation}
{\mathcal M} (z) \overset{\text{law}}{=} z \,.
\end{equation}
In particular, if $\nu_{\mu}(\d z)$ has a density, say $f_{\mu}(z)$, then this equality in law translates into the following equation for $f_{\mu}$:
\begin{equation}
f_{\mu}(z) = \int_{\text{SL} \left ( 2, {\mathbb R} \right )}   \mu (\d M ) \left ( f_{\mu} \circ {\mathcal M}^{-1} \right ) (z)
\frac{\d {\mathcal M}^{-1}}{\d z}(z)\,.
\label{furstenbergEquation}
\end{equation}
where, with some abuse of notation, $z$ is no longer random. Knowing $f_{\mu}$, the Lyapunov exponent $\gamma_{\mu}$ of the product $\Pi_n$
may then be computed by the formula
\begin{equation}
\gamma_{\mu} = \int_{{\mathbb R}} \d z f_{\mu} ( z ) \int_{\text{SL} \left (2,{\mathbb R} \right )}  \mu (\d M) \ln \frac{\left | M \left(\begin{array}{c} z \\ 1 \end{array}\right) \right |}{ \left |  
\left(\begin{array}{c} z \\ 1 \end{array}\right) \right |} \,.
\label{furstenbergFormula}
\end{equation}

The calculation of $\gamma_{\mu}$ when $M$ is of the form (\ref{iwasawaDecomposition}) and $\theta$ is exponentially distributed with mean $1/p$, independent
of $w$ and $u$, may then be related to the calculation of $\gamma$  in the previous subsection as follows:
\begin{enumerate}
\item The projective variable $z$ is the Riccati variable $Z$ with $E=k^2=1$.
\item The Furstenberg equation (\ref{furstenbergEquation}) for the invariant density reduces to the Frisch--Lloyd equation (\ref{frischLloydEquation})
for the stationary density.
\item The Furstenberg formula (\ref{furstenbergFormula})  for $\gamma_{\mu}$ reduces to Formula (\ref{riccatiFormula}) for $\gamma$ and
  \begin{equation}
\gamma_{\mu} = \frac{1}{p}\,\gamma\,.
\end{equation}
The factor $1/p$ is due to the difference between the definitions:
whereas (\ref{gamma}) involves the distance variable $x$, the growth of the matrix product is measured with respect to the index $n$ in (\ref{gammaMu})--- $x$ and $n$ being related by 
(\ref{RelationXN}).
\end{enumerate}

\subsection{Known solvable cases involving random point-scatterers}
\label{exampleSubsection}

There is no systematic method for solving the integro-differential
equation (\ref{frischLloydEquation}), but some solvable cases have been found and are listed in
Table \ref{tab:SolvableCases}.  All but the first of these cases have in common that one or both the
parameters $w$ and $u$ in the expression for the matrix $B$ are
exponentially or gamma distributed. The density of the exponential and gamma distributions 
satisfies a differential equation with {\em constant} coefficients,
and this leads
to a trick for reducing
(\ref{frischLloydEquation}) to a purely differential 
form. This trick will be illustrated in the next section.

Here, we consider the simplest example and merely write down the result.
This example consists of taking
\begin{equation}
B = \left(\begin{array}{cc}
1 & u \\
0 & 1
\end{array}\right)
\end{equation}
where $u$ is exponentially distributed with mean $1/q$. Then the stationary density solves
\begin{equation}
\frac{\d}{\d z} \left [ \left ( z^2+k^2 \right ) f(z) \right ] - p\, f(z) + q  \left ( z^2+k^2 \right ) f(z) = q N\,.
\end{equation}
The constant $N=N(E)$ appearing on the right-hand side is the {\it
  integrated density of states} per unit length of the quantum model; the {\it Rice formula}
\begin{equation}
N(E)  = \lim_{|z| \rightarrow \infty} z^2 f(z)
\label{riceFormula}
\end{equation}
expresses its relationship
to the tail of the stationary density. After integration, we obtain
\begin{equation*}
f(z) = \frac{q N}{z^2+k^2} \exp \left [- q z + \frac{p}{k}\arctan \frac{z}{k}  \right ] \int_{-\infty}^z \exp \left [q t - \frac{p}{k} \arctan \frac{t}{k}  \right ] \,\d t\,.
\end{equation*}
The integrated density of states is then determined by the requirement that $f$ be a probability density. The final formulae for $N$ and $\gamma$
may be expressed neatly via the  characteristic function (\ref{characteristicFunction}) of the disordered system:
for $E =k^2>0$,
\begin{equation}
\Omega (E) = 2 \i k\, \frac{W_{\frac{-\i p}{2 k},\frac{1}{2}}' \left ( - 2 \i k q \right )}{W_{\frac{-\i p}{2 k},\frac{1}{2}}\left ( - 2 \i k q \right )}
\label{nieuwenhuizenFormula}
\end{equation}
where $W_{\alpha,\beta}$ is a Whittaker function~\cite{GR}.
This formula, which was discovered by Nieuwenhuizen \cite{Ni} using a different approach,
corresponds here to the case of ``repulsive'' scatterers ($u > 0$); it 
is easily adapted to the
``attractive'' case $u<0$ by simply changing the sign of $q$ 
in Eq.~(\ref{nieuwenhuizenFormula}) \cite{CTT1}.

\begin{table}[!ht]
  \centering
  \begin{tabular}{lllll}
      Density of $\ell$ & Density of $u$ & Density of $w$ & Reference & Special function\\
      \hline
      $\delta(\ell- 1/p)$ 
           & $\frac{q/\pi}{q^2+u^2}$ 
           & $\delta(w)$ 
           & \cite{Is} 
           & \\
      $p\,\e^{-p\ell}\,\mathbf{1}_{\mathbb{R}_+}(\ell)$
           & $q\,\e^{-qu}\,\mathbf{1}_{\mathbb{R}_+}(u)$  
           & $\delta(w)$ 
           & \cite{Ni,CTT1} 
           & Whittaker \\
      $p\,\e^{-p\ell}\,\mathbf{1}_{\mathbb{R}_+}(\ell)$
           & $q^2 u\,\e^{-qu}\,\mathbf{1}_{\mathbb{R}_+}(u)$
           & $\delta(w)$ 
           & \cite{Ni,CTT1}
           & Whittaker \\
      $p\,\e^{-p\ell}\,\mathbf{1}_{\mathbb{R}_+}(\ell)$
           & $\delta(u)$ 
           & $q\,\e^{-qw}\,\mathbf{1}_{\mathbb{R}_+}(w)$  
            & \cite{CTT1}
            &Hypergeometric\\
      $p\,\e^{-p\ell}\,\mathbf{1}_{\mathbb{R}_+}(\ell)$
           & $\delta(u)$ 
           & $q\,\e^{-2q|w|}$  
           & this article
           & generalized Hypergeometric\\
      \hline
  \end{tabular}
  \caption{\it Known solvable cases involving random point scatterers.
  The notation $\mathbf{1}_{A}(x)$ is used for the function that equals $1$ if $x\in A$ and $0$ otherwise.}
  \label{tab:SolvableCases}
\end{table}
Besides the Lyapunov exponent, the low-energy behaviour of the
density of states is also of interest; it has been analysed for point scatterers with random
positions in the following cases:
({\it i}) when $w=0$ and the probability density of $u$ is supported on the positive half-line but is otherwise arbitrary \cite{Ko,LGP};
({\it ii}) in the converse situation where $u=0$ and it is $w>0$ that is random~\cite{CTT2}.

\subsection{Continuum limit and models involving Gaussian white noises
or L\'evy noises}

No exact solution is known in the case where {\em all three} parameters
$\ell = \theta$, $w$ and $u$ are random.
Some insight into the general case may however be gained by considering the continuum
limit of an infinitely dense set of point scatterers with vanishingly weak strengths, i.e.
\begin{equation*}
\theta\,,\;w\,,\;u \rightarrow 0\,.
\end{equation*}
This limit was studied very recently in Ref.~\cite{ComLucTexTou12} in the most general case where the three parameters of the Iwasawa
decomposition (\ref{iwasawaDecomposition}) are correlated--- the
matrices $M_j$ being still mutually independent.
This recent work corresponds to a generalisation of the study of the
Schr\"odinger equation with a potential 
$
V(x)=\eta(x)+W(x)^2-W'(x)
$
that combines two (possibly correlated) Gaussian white
noises $\eta$ and $W$~\cite{HagTex08}; see
Ref.~\cite{ComLucTexTou12} for a detailed discussion. 

Finally let us mention a generalisation of the model in another
direction: by letting the density of
scatterers tend to infinity, it is possible to
study certain singular limiting measures for the strength of the scatterers.
Such cases may be described in terms of L\'evy noises and were studied in \cite{CTT2,Ko}.


\section{Supersymmetry: a new solvable case}
\label{susySection}

In this section, we shall restrict our attention to the particular
case where the point scatterers are described by
transfer matrices of the form
\begin{equation}
  B_j = \left(\begin{array}{cc}
  \e^{w_j} & 0 \\
  0 & \e^{-w_j}
  \end{array}\right)
  \,.
  \label{eq:SusyScatterer}
\end{equation}
As shown in Ref.~\cite{CTT1}, 
this corresponds to considering the potential 
\begin{equation}
V = W^2 - W'
\label{supersymmetricPotential}
\end{equation}
where $W$ is a superposition of delta-scatterers
\begin{equation*}
W(x)=\sum_jw_j\,\delta(x-x_j)\,.
\end{equation*}
A potential of the form (\ref{supersymmetricPotential}) is said to be {\it supersymmetric} with {\it superpotential} $W$.
Hence we shall refer to the point scatterer associated with (\ref{eq:SusyScatterer}) as ``supersymmetric''.
This case has many
interesting features: the corresponding Hamiltonian is factorisable and, when the superpotential is deterministic and has the so-called ``shape-invariance'' property,
the discrete spectrum may be obtained exactly~\cite{CKS}. 

The disordered case is of interest 
in several physical contexts \cite{ComTex98}: in
particular, the Schr\"odinger equation can be mapped onto the Dirac
equation with a 
random mass--- a model relevant in several contexts of condensed matter
physics; see the reviews \cite{ComTex98,Gog82,LGP} and the introduction in
Ref.~\cite{TexHag10}. 
The model is also closely 
connected to
diffusion in a random environment \cite{BCGL}; indeed, if we set
\begin{equation*}
\psi (x) = \exp \left [ -\int W(x) \,\d x \right ]\, U(x) 
\end{equation*}
then
\begin{equation}
\frac{1}{2} \,U''(x) - W(x) \,U'(x) = \lambda \,U (x) \quad
\text{with}\;\;\lambda 
= -\frac{E}{2}\,.
\label{infinitesimalGenerator}
\end{equation}
When $E<0$, this is the equation satisfied by the Laplace transform of a {\it hitting time} of the diffusion in the environment $\int W(x) \,\d x$; see for instance~\cite{Ca}.

After integrating with respect to $z$, the Frisch--Lloyd equation
(\ref{frischLloydEquation}) for the stationary density becomes 
\begin{equation}
\left ( z^2 + E \right ) f(z) + p \int_{\mathbb R} \d w \,\varrho (w) \int_{z}^{z\e^{-2 w}} f(t)\,\d t = N(E) 
\label{supersymmetricFrischLloyd}
\end{equation}
where $\varrho$ is the density of the random variable $w$.  Our strategy for computing the Lyapunov exponent will make use of three simple observations: 
firstly,
the equivalence between the supersymmetric Schr\"{o}dinger equation and Equation (\ref{infinitesimalGenerator}) implies that
the spectrum is necessarily contained in ${\mathbb R}_+$. It follows
that $N(E) =0$ for $E<0$ and that the stationary density $f$
is supported on the positive half-line; this facilitates the calculation of the Lyapunov exponent. The second observation is that
the characteristic function (\ref{characteristicFunction})
is an {\it analytic function of $E$}, except for a branch cut along the positive half-axis. Therefore, if we find an analytical formula
for the Lyapunov exponent when $E<0$, analytic continuation will furnish a formula for the case $E>0$. The third observation ---which we alluded to earlier and whose exploitation goes back to the
works of Nieuwenhuizen \cite{Ni} and Gjessing and Paulsen \cite{GP}--- is that the integral term appearing in (\ref{supersymmetricFrischLloyd})
may be eliminated if $\varrho$ solves a linear differential equation with (piecewise) constant coefficients.

With these points in mind, let us set
\begin{equation*}
E = -k^2<0
\end{equation*}
and take
\begin{equation}
\varrho(w) = q \,e^{-2q |w|}\,.
\label{laplaceDensity}
\end{equation}
If we write 
$f_{p,q,k}$ to indicate the dependence of the stationary density $f$ on the parameters $p$, $q$ and $k$, then we have the following identity:
\begin{equation}
f_{p,q,k} (z) = \frac{1}{k}\,f_{\frac{p}{k},q,1} \left ( \frac{z}{k} \right )\,.
\label{negativeEnergyIdentity}
\end{equation}
It is therefore sufficient to consider the case $k=1$, i.e. $E=-1$.
It is then easily deduced from the evenness of $\varrho$ that
\begin{equation}
f(z) = \frac{1}{z^2} f(1/z)\,.
\label{involutionProperty}
\end{equation}
As a consequence, to know $f(z)$ for $z\geq1$ is to know $f(z)$ for $0<z<1$, and vice-versa; we shall at times implicitly make use of this helpful
property.

\subsection{Reduction to a differential equation}
\label{reductionSubsection}
Following the recipe outlined in \cite{CTT1}, we shall now deduce from the Frisch--Lloyd equation (\ref{supersymmetricFrischLloyd}) a {\it differential}
equation for the density $f$. To this end, we introduce the kernel
\begin{equation}
K (y) = -\text{sign}(y) \,\frac{1}{2} \e^{-2q |y|}
\end{equation}
and note, for future reference, that $K$ satisfies
\begin{equation}
K'(y) = - \text{sign}(y) \,2 q\, K(y)\,, \;\; y \ne 0\,,
\label{kernelDifferentialEquation}
\end{equation}
subject to
\begin{equation}
K(0 \pm) = \mp \frac{1}{2}\,.
\label{kernelConditions}
\end{equation}
After permuting the order of integration in the Frisch--Lloyd equation, we find
\begin{equation}
(z^2-1) f(z) + p \int_0^\infty K \left ( \frac{1}{2} \ln \frac{z}{t} \right ) f(t)\,\d t = 0\,.
\label{permutedFrischLloydEquation}
\end{equation}
Set
\begin{equation*}
\varphi(z) := (z^2-1) f(z)
\end{equation*}
and write
$$
\int_0^\infty K \left ( \frac{1}{2} \ln \frac{z}{t} \right ) f(t)\,\d t = \int_0^z K \left ( \frac{1}{2} \ln \frac{z}{t} \right ) f(t)\,\d t +
\int_z^\infty K \left ( \frac{1}{2} \ln \frac{z}{t} \right ) f(t)\,\d t\,.
$$
Differentiation of (\ref{permutedFrischLloydEquation}) with respect to $z$ then yields:
\begin{eqnarray}
\nonumber
\varphi'(z) + p \,K(0+) f(z) - p \,\frac{q}{z} \,\int_0^z K \left ( \frac{1}{2} \ln \frac{z}{t} \right ) \,f(t)\,\d t \\
- p \,K(0-) f(z) + p \,\frac{q}{z} \,\int_z^\infty K \left ( \frac{1}{2} \ln \frac{z}{t} \right ) \,f(t)\,\d t 
= 0\,.
\end{eqnarray}
Hence, in view of (\ref{kernelConditions}),
$$
z\,\left [ \varphi'(z) - p \,f(z) \right ] - p \,q \,\int_0^z K \left ( \frac{1}{2} \ln \frac{z}{t} \right ) \,f(t)\,\d t 
+ p \,q\,\int_z^\infty K \left ( \frac{1}{2} \ln \frac{z}{t} \right ) \,f(t)\,\d t 
= 0\,.
$$
Another differentiation leads, after use of (\ref{kernelConditions}), to
\begin{equation}
\frac{\d}{\d z} \left \{ z\,\left [ \varphi'(z) - p \,f(z) \right ] \right \} + p\,\frac{q^2}{z} \,\int_0^\infty 
K \left ( \frac{1}{2} \ln \frac{z}{t} \right ) \,f(t)\,\d t = 0\,.
\end{equation}
The integral term may then be eliminated by making use of Equation (\ref{permutedFrischLloydEquation}); the result is the second-order
linear differential equation
\begin{equation}
z \frac{\d}{\d z} \left [ z\,\left ( \varphi' - p \,\frac{\varphi}{z^2-1} \right ) \right ] -q^2 \varphi = 0\,.
\label{differentialEquation}
\end{equation}

\subsection{The case $q=1$}
\label{qIsOneSubsection}

It may be verified by direct substitution that, when $q=1$, the differential equation (\ref{differentialEquation}), expressed
in terms of the unknown $f$, has the solution  
\begin{equation}
  f_{p,1,1}(z) = \frac{A}{p-2} \, 
  \frac{1}{z} 
  \left [ 
     1- \frac{1+p z + z^2}{|z-1|^{1-\frac p2}(z+1) ^{1+\frac p2}} 
  \right ]\,
\label{qIsOneDensityFormula}
\end{equation}
for $p\neq2$. 
The normalisation constant $A$ satisfies
\begin{eqnarray}
\nonumber
1 = \frac{2A}{p-2} \int_0^1 \left  [ 1- \left ( \frac{1-z}{1+z} \right
  )^{\frac{p}{2}} \frac{1+p z + z^2}{1-z^2} \right ] \,\frac{\d z}{z}
\\
\overset{\underset{\downarrow}{z=\frac{1-x}{1+x}}}{=}
\frac{2A}{p-2} \int_0^1 \left [ 2 \frac{1-x^{\frac{p}{2}-1}}{1-x^2} + \left (1-\frac{p}{2} \right ) x^{\frac{p}{2}-1} \right ]\,\d x \\
= \frac{2A}{p-2} \left \{ 2 \int_0^1  \frac{1-x^{\frac{p}{2}-1}}{1-x^2} \,\d x + \frac{2}{p} \left ( 1-\frac{p}{2} \right ) \right \}
\,.
\end{eqnarray}
The value of this last integral is given explicitly by \cite{GR}, \S 3.269, in terms of the digamma function $\digamma$. Hence
\begin{equation}
\frac{1}{A} = \frac{1}{2} \frac{\digamma (p/4)-\digamma(1/2)}{p/4-1/2} - \frac{2}{p} \,.
\label{denominator}
\end{equation}
The case $p=2$ may be studied by letting $p\to2$ in
 Equations (\ref{qIsOneDensityFormula},\ref{denominator}). We get
 \begin{equation}
 f_{2,1,1}(z)=A
   \left[
    \frac1{2z} \ln\left|\frac{z+1}{z-1}\right| - \frac1{(z+1)^2}
  \right]
  \hspace{0.5cm}\mbox{and}\hspace{0.5cm}
  \frac {1}{A} = \frac{\pi^2}{4}-1\:.
\end{equation}
Plots of the density (\ref{qIsOneDensityFormula}) for several values of
$p$ are shown in Figure~\ref{fig:densities}. 
\begin{figure}[!ht]
  \centering
  \includegraphics[scale=1]{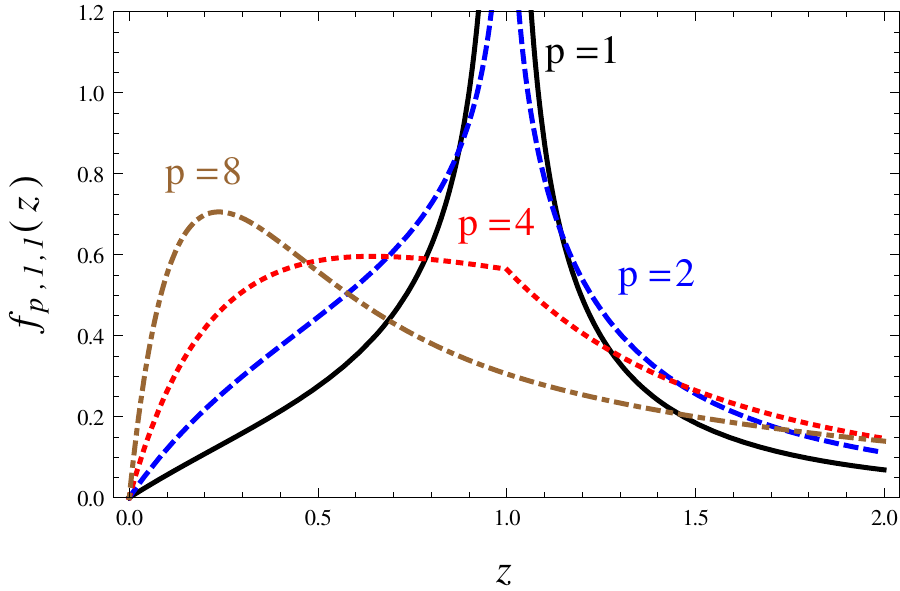}
  \caption{\it Plots of the stationary probability density $f(z)$
    for $q=1$ and energy $E=-1$.}
\label{fig:densities}
\end{figure}
The integral on the right-hand side of Formula (\ref{riccatiFormula}) may also be computed exactly; the end result is
\begin{equation}
  \Omega (-k^2) = \frac p2\,
  \frac{1 -\frac{4 k^2}{p^2}}
          {\digamma\left(\frac{p}{4k}\right)  -
            \digamma\left(\frac{1}{2}\right) +\frac{2k}{p} -1}\,.
\label{negativeEnergyCharacteristicFunction}
\end{equation}
Analytic continuation to positive energies $E=k^2>0$ consists of replacing $k$ by $-\i k$; this yields 
\begin{equation}
  \Omega(E + \i 0+) = \frac p2\,
  \frac{1 +\frac{4 k^2}{p^2}}
          {\digamma\left(\frac{\i p}{4k}\right)  -
            \digamma\left(\frac{1}{2}\right) -\frac{2\i k}{p} -1}\,.
  \label{positiveEnergyCharacteristicFunction}
\end{equation}
Various limits may be analysed with the help of these expressions.
Using 
\begin{equation*}
\digamma(z)\underset{z\to0}{=}-\frac1z-\mathbf{C}+\frac{\pi^2}{6}z+\mathcal{O}(z^2)
\end{equation*}
where $\mathbf{C} = 0.577 \ldots$ is the Euler-Mascheroni constant,
we deduce from (\ref{negativeEnergyCharacteristicFunction})
\begin{equation}
\Omega (-k^2)
  \underset{k\to\infty}{=}k+p\,\left (\ln2-\frac12 \right )+\mathcal{O}(k^{-1})\,.
\end{equation}
Analytic continuation then shows that $N(k^2)\sim k/\pi$
in the limit  $k\to\infty$, as expected
from the free case. Also,  the high-energy Lyapunov exponent
tends to a constant: $\gamma(k^2)\sim p\,(\ln2-1/2)$; see Fig.~\ref{fig:IDosLyap}.

In the low-energy limit, by using
\begin{equation*}
\digamma(z)\underset{z\to\infty}{=}\ln z-\frac1{2z}+\mathcal{O}(z^{-2})\,,
\end{equation*}
we obtain 
\begin{equation}
\Omega (-k^2) \underset{k\to0}{=}
  \frac{p/2}{\ln(\frac{p}{k}) +\mathbf{C}-1}
  +\mathcal{O}\left(\frac{k^2}{\ln k}\right)
  \:.
\end{equation}
Analytic continuation to positive energies allows one to recover the
Dyson singularity of the integrated density of states and the corresponding vanishing of the
Lyapunov exponent:
\begin{equation*}
  N(k^2)\sim\frac{g}{2\ln^2k}\;\;\text{and}\;\;\gamma(k^2)\sim\frac{g}{|\ln k|}
\end{equation*}
with $g=p/2$. 
In this $k\to0$ limit, we thus recover as expected the same results as
in the case where the superpotential $W$ is a Gaussian white noise;
the characteristic function for this case
is recalled below.
The fact that $\gamma \rightarrow 0+$ as $E \rightarrow 0+$
means that the eigenstates become delocalised (or extended) in that
limit; see the discussion in Section~\ref{conclusionSection}.
Plots of $\gamma$ and $N$ for positive energy are shown in Figure~\ref{fig:IDosLyap}. 
\begin{figure}[!ht]
  \centering
  \includegraphics[scale=0.825]{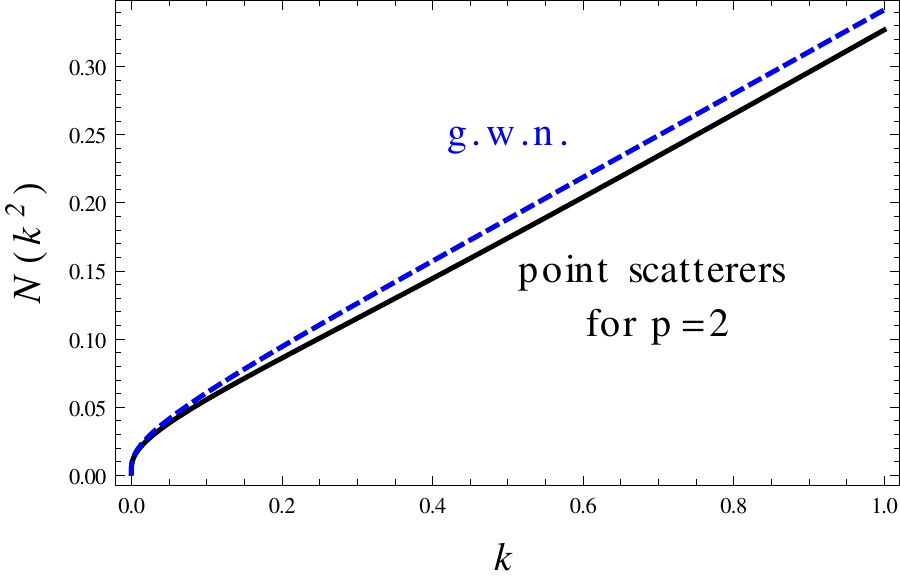}
  \hspace{0.25cm}
  \includegraphics[scale=0.825]{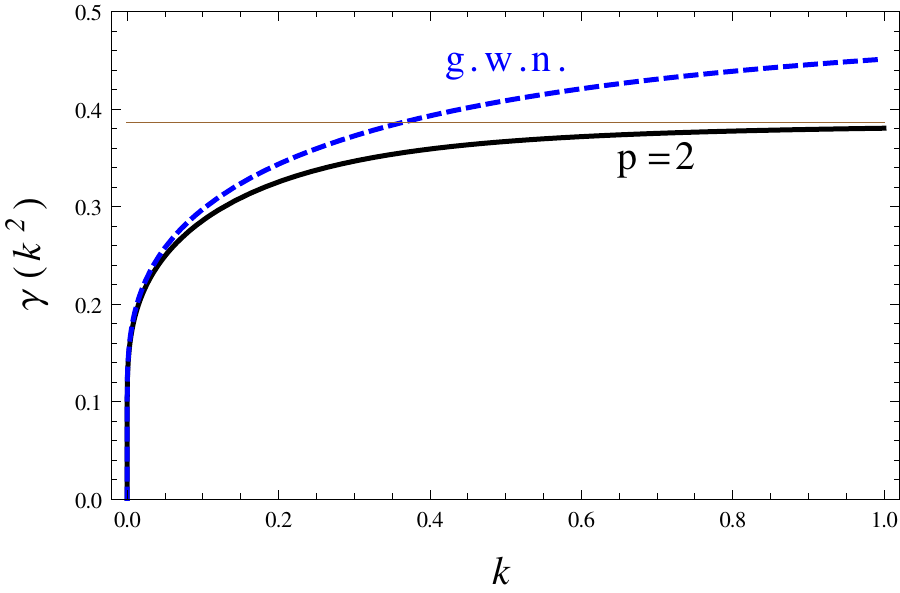}
\caption{\it Plots of $N$ and  $\gamma$ against $k=\sqrt{E}$ for $p=2$
  and $q=1$. The
  integrated density of states exhibits the Dyson singularity. For comparison,
 the results corresponding to the case where $W$ is a Gaussian white noise
  of strength $g=p/2q^2=1$ are also shown as blue dashed curves. The
  thin line is the asymptote of the Lyapunov exponent.}
\label{fig:IDosLyap} 
\end{figure}

\subsection{The characteristic function for arbitrary $q>0$}
\label{arbitrarySubsection}

Although the calculations are somewhat tedious, the more general case $q >0$ is also tractable. The results are summarised in the
following 
\begin{theorem}
The invariant probability density is given by
\begin{equation*}
f_{p,q,1}(z) = \frac{C}{z} \int_0^\infty \e^{-\frac{x}{4} \left (z
    +1/z \right )} M_{\frac{p}{4},q}(x) \,\frac{\d x}{\sqrt{x}} \,.
\end{equation*}
where $C$ is a normalisation constant and $M_{\alpha,\beta}$ is a Whittaker function. 
\label{whittakerTheorem}
\end{theorem}

\begin{corollary}
For $E = k^2 >0$.
\begin{equation*}
\Omega(E + \i 0+) = -\i  k \frac{q+1}{q} 
\frac{{_3}F_2 \left ( q+\frac{1}{2}-\frac{\i  p}{4k},q+2,q;\,2q+1,q+\frac{3}{2};\,1 \right )}{{_3}F_2 \left ( q+\frac{1}{2}-\frac{\i  p}{4k},q+1,q+1;\,2q+1,q+\frac{3}{2};\,1 \right )}\,.
\end{equation*}
\label{whittakerCorollary}
\end{corollary}

\begin{proof}
See \ref{proofAppendix}.
\end{proof}

Although this formula for $\Omega$ in terms of a known special function is pleasing, the task of extracting from it concrete
information is not entirely straightforward. We proceed to discuss various limits.

The integrated density of states behaves in the limit $E\to+\infty$ as in the free case. 
We may obtain the  behaviour of $\gamma$ by using the perturbative
approach described in Ref.~\cite{BieTex08}. 
When the superpotential $W$ is given by
a superposition of delta--functions, the high energy Lyapunov exponent of
the supersymmetric Hamiltonian tends to a finite limit given by 
$\gamma_\infty = p\,\mathbb{E}(\ln\cosh w)$.
Therefore we obtain 
\begin{equation}
  \gamma(E) \sim \frac p2
  \left[ \frac1q - \digamma\left(\frac{q}{2}+1\right) 
     + \digamma\left(\frac{q+1}{2}\right) \right] 
  \hspace{1cm}
  \mbox{as } E\to\infty
  \:.
\end{equation}
For example, in the case $q=1$, we recover the limit  $\gamma \sim  p\,(\ln2-1/2)$
discussed earlier.

The low-energy limit may be conveniently studied by using the integral
representation of the characteristic function
\begin{equation}
  \label{eq:IntegralRepresentation}
  \Omega(-k^2) = k\, 
  \frac{ \displaystyle
    \int_0^\infty\frac{\d x}{\sqrt{x}}\,M_{\frac{p}{4k},q}(x)\,K_1(x/2) 
      }{\displaystyle
    \int_0^\infty\frac{\d x}{\sqrt{x}}\,M_{\frac{p}{4k},q}(x)\,K_0(x/2)
      }
\end{equation}
following from (\ref{integralFormula}). We set $E=-k^2=-1$ for simplicity. 
The $p\to\infty$ limit 
provides the $E\to0$
behaviour, after reintroducing $k$.
We now analyse the integrals in this limit.
There holds~\cite{GR}
\begin{equation}
M_{\frac{p}{4},q}(x)\sim x^{q+1/2}\;\; \text{as $x \rightarrow 0$}
\label{eq:AsympM}
\end{equation}
and,
since the Whittaker function solves the differential equation
\begin{equation*}
w''(x)=\frac{x^2-p\,x+(4q^2-1)}{4x^2}\,w(x)\,,
\end{equation*}
we expect the asymptotic form (\ref{eq:AsympM}) to hold for 
$x$ smaller than $(4q^2-1)/p$ 
(this reasoning assumes that $4q^2-1>0$). 
Moreover, we see from the differential equation that the Whittaker
function is a rapidly oscillating function in the interval 
$[x_-,x_+]$ where
$x_\pm=\frac{p}{2}\pm\sqrt{(\frac{p}{2})^2-4q^2+1}$;
in the limit of interest here we have 
$x_-\sim (4q^2-1)/p$ and $x_+\sim p$. 
Thus we expect that the contribution to the integrals of the interval
$[x_-,x_+]$ is negligible.
For $x$ exceeding $x_+$ we have 
\begin{equation*}
\frac{1}{\sqrt x} \,  M_{\frac{p}{4},q}(x)\,K_s(x/2)\sim x^{-\frac p4-1}
\:.
\end{equation*}
Finally both integrals in (\ref{eq:IntegralRepresentation}) are
dominated by the interval $[0,x_-]$, from which we get
\begin{equation}
  \label{eq:LimitBehaviour}
  \Omega(-1) \underset{ p\to\infty }{ \sim }
  \frac{q+1}{q} 
  \frac{ 2 }{x_-\ln(1/x_-) }
  \sim \frac{p}{\ln p}
  \:.
\end{equation}
Re-introducing $k$, we indeed get the expected behaviour
$\Omega(-k^2)\sim-1/\ln k$ for $k\to0$.

The limit $p\to\infty$ and
$q\to\infty$ keeping $p/q^2$ fixed corresponds to a case where 
$W(x)$ converges to a Gaussian white noise of strength
$g=p\,\mathbb{E}(w^2)=p/2q^2$--- a model that has been solved in
Ref.~\cite{BCGL,OvcEri77}  
(see also \cite{LGP}). Therefore we expect that the characteristic
function (\ref{eq:IntegralRepresentation}), i.e. Eq.~(\ref{eq:OmegaGeneralCase}),
converges in this limit towards the expression
\begin{equation}
\Omega(-k^2) 
\underset{p \to \infty\,, \;q=\sqrt{p/2g}}{\longrightarrow} 
k \frac{K_1(k/g)}{K_0(k/g)}
\:.
\end{equation}


\section{Concluding remarks: limitations of the Lyapunov exponent for the study of localisation}
\label{conclusionSection}

To close this short review, we point out the limitations of the Lyapunov exponent as a means of 
characterising the localised nature of the spectrum in a disordered system. 
A first observation is that the definition
(\ref{gamma}) relies on the assumption of a linear increase of
$\ln|\psi(x;E)|$, i.e. on the application of the central limit theorem. It is
however possible to consider models of disorder such that
$\ln|\psi(x;E)|$ scales as $x^\alpha$, leading to the  phenomenon  of
sub-localisation (for $\alpha<1$) or super-localisation
(for $\alpha>1$); this has been studied in, for example,  
Refs.~\cite{BieTex08,BooLuc07,Luc05,RusKanBunHav01} (see also references therein). 
However,
this limitation is easily overcome by extending the
notion of the Lyapunov exponent
so that it characterises
the typical length over which eigenstates are localised. 

Another, more serious, limitation of the usefulness of the Lyapunov exponent pertains
to the possible existence of extended states in one dimension.
This might occur when the disordered potential is correlated \cite{IzrKroMak12}--- as happens
for instance in the dimer model. It might occur also 
for symmetry reasons, as is the case for the supersymmetric Hamiltonians
studied in Refs.~\cite{BCGL,CDM,ComTex98,HagTex08,TexHag10} 
and in section~\ref{susySection};
there, the Lyapunov exponent $\gamma(E)$ decays to zero in the limit of zero energy.
In this respect,
it must first be pointed out that the vanishing of the Lyapunov
exponent is a necessary but not a sufficient condition for
delocalisation; sometimes it only signals
sub-localisation.
Secondly, several works on disordered supersymmetric
quantum mechanics have shown that
the Lyapunov exponent, when it vanishes at the bottom of the spectrum 
does not provide the relevant
length scale characterising the localisation of eigenstates at low non-zero energies.
As mentioned already in
section~\ref{susySection}, when
$\mathbb{E}(W)=0$, the Lyapunov exponent vanishes like  
$\gamma(E)\sim g/|\ln E|$ as $E\to0$
\cite{BCGL,CDM,ComTex98}. This suggests that low-energy
eigenstates are localised on a scale commensurate with $|\ln E|/g$.
However, the study of the ordered statistics problem \cite{Tex00,TexHag10}, as well as the analysis of other physical quantities such as the averaged
Green's function at non coinciding points~\cite{Gog82,BCGL} or the
boundary-sensitive averaged density of states \cite{TexHag10}, show
that eigenstates are localised on a much larger scale, namely
$|\ln E|^2/g$.

The inability of the Lyapunov exponent to capture the localisation
property for energies $E\sim E_c$ where $\gamma(E_c)=0$ may be
related to the fact that it is defined, via Eq.~(\ref{gamma}), in terms of a
solution $\psi(x;E)$ of the Cauchy problem that vanishes at just {\it one} boundary. 
This definition may indeed fail to describe the
interesting properties of the real eigenstates--- which are solutions of the
Schr\"odinger equation vanishing at the {\em two}
boundaries $x=0$ and $x=L$--- as these eigenstates become less
localised.

\section*{Acknowledgements}

This work was supported by {\it Triangle de la Physique} under Grant 2011-016T-1DDisQuaRM.


\appendix

\section{Proofs of Theorem \ref{whittakerTheorem} and its corollary}
\label{proofAppendix}

Let us look for a solution of the differential equation (\ref{differentialEquation}) of the form
\begin{equation*}
\varphi(z) := z^q \left ( \frac{1-z}{1+z} \right )^{\frac{p}{2 }} \,y(z)\,.
\end{equation*}
Then
\begin{equation}
y'' + \left [ \frac{2 q +1}{z} +  \frac{p/2}{z-1}+\frac{-p/2}{z+1}  \right ] y' + \frac{p\,q}{z (z-1)(z+1)}\, y = 0\,. 
\label{reducedEquation}
\end{equation}
The solution we require has various representations in terms of special functions.

\subsection{Heun's function}
\label{heunSubsection}

Heun's differential equation is \cite{Ro}
\begin{equation}
y'' + \left [ \frac{\eta}{z} + \frac{\delta}{z-1} + \frac{\varepsilon}{z-a} \right ] y' + \frac{\alpha\,\beta z -q}{z (z-1)(z-a)}\, y = 0
\label{heunEquation}
\end{equation}
where the parameters satisfy the relation
\begin{equation*}
\alpha + \beta + 1 = \eta + \delta + \varepsilon \,.
\end{equation*}
Heun's function is the particular solution analytic inside the unit disk:
\begin{equation}
Hl (a,q;\alpha,\beta,\eta,\delta;z) := \sum_{n=0}^\infty y_n z^n\,, \;\;|z|<1\,,
\label{heunFunction}
\end{equation}
where the $y_n$ satisfy the recurrence relation
\begin{eqnarray}
\nonumber
a (n+2) ( n+1+\eta) \,y_{n+2} \\ 
\nonumber
- \left [ (n+1)(n+\eta+\delta) a + (n+1)(n+\eta+\varepsilon) + q \right ] \,y_{n+1}  \\
+ (n+\alpha)(n+\beta) \,y_n = 0
\label{heunRecurrenceRelation}
\end{eqnarray}
with the starting values
\begin{equation}
y_0 = 1 \;\;\text{and}\;\; y_1 = \frac{q}{a \eta}\,.
\label{heunStartingValues}
\end{equation}
Comparing Equation (\ref{reducedEquation}) with Heun's equation (\ref{heunEquation}), we find
the particular solution
\begin{equation*}
y = Hl ( -1,-pq ; 0,2 q, 2q+1,p/2;z)\,.
\end{equation*}
The recurrence relation for the coefficients $y_n$ of this solution is
\begin{equation}
(n+2)(n+2q+2) \,y_{n+2} = p (n+q+1) \,y_{n+1} + n(n+2 q) \,y_n
\label{recurrenceRelation}
\end{equation}
with the starting values
\begin{equation}
y_0 = 1 \;\;\text{and}\;\; y_1 = \frac{p q}{2 q +1}\,.
\label{startingValues}
\end{equation}
It follows in particular that $y$ is strictly positive for $z>0$. Hence, for $z \in [0,1)$,
\begin{equation}
f_{p,q,1} (z) = \frac{C z^q}{1-z^2} \left ( \frac{1-z}{1+z} \right )^{\frac{p}{2}} Hl ( -1,-pq ; 0,2 q, 2q+1,p/2;z)
\label{heunDensity}
\end{equation}
where $C$ is a normalisation constant.
The property (\ref{involutionProperty})  provides the obvious formula
in the interval $z>1$.

\subsection{Gauss' hypergeometric function}
\label{hypergeometricSubsection}

Set
\begin{equation*}
a := \frac{p}{4} + q - \frac{1}{2}\,, \;\; b := \frac{p}{4} \;\;\text{and}\;\;c := q + \frac{1}{2}\,.
\end{equation*}
Then, by comparing the Taylor series, we find that
\begin{eqnarray}
Hl ( -1,-pq ; 0,2 q, 2q+1,p/2;z) \nonumber\\ 
= \frac{(1+z)^{2b}}{c(c+1)} 
\left \{ (a+1)(b+1) z^2 (1-z^2)\,{_2}F_1 \left (a+2,b+2;c+2;z^2 \right) \right . 
\nonumber\\
\left . + (c+1) \left [ c (1-z) -b z \right ] (1+z)\,
  {_2}F_1 \left(a+1,b+1;c+1;z^2 \right) \right \}\,. 
\end{eqnarray}
The contiguity relation
\begin{eqnarray}
t (1-t) (a+1)(b+1) \,{_2}F_1(a+2,b+2;c+2;t) \nonumber\\ 
+ \left [ c-(a+b+1) t \right ] (c+1) \,{_2}F_1(a+1,b+1;c+1;t) \nonumber\\
 - c(c+1) \,{_2}F_1(a,b;c;t) = 0 
\end{eqnarray}
leads to the simpler expression
\begin{eqnarray}
Hl ( -1,-pq ; 0,2 q, 2q+1,p/2;z) \\\nonumber
 = (1+z)^{2b} \left \{ {_2}F_1 \left ( a,b;c;z^2 \right )
- \frac{b}{c} z (1-z)\,{_2}F_1 \left (a+1,b+1;c+1;z^2 \right ) \right \}\,.
\end{eqnarray}
The differentiation formula
\begin{equation*}
{_2}F_1 '(a,b;c;t) = \frac{a b}{c} {_2}F_{1}(a+1,b+1;c+1;t)
\end{equation*}
then yields
\begin{eqnarray}
&&\hspace{-0.5cm}
Hl ( -1,-pq ; 0,2 q, 2q+1,p/2;z) \nonumber\\
 &=& (1+z)^{2b} \left \{ {_2}F_1 \left ( a,b;c;z^2 \right )
- \frac{z}{a} (1-z)\,{_2}F_1 '\left (a,b;c;z^2 \right ) \right \}  \\
 &=& (1+z)^{2b} \left \{ {_2}F_1 \left ( a,b;c;z^2 \right )
 - \frac{1}{2 a} (1-z)\, \frac{\d}{\d z} {_2}F_1 \left (a,b;c;z^2 \right ) \right \} \\
&=& \frac{-1}{2a} \frac{(1+z)^{2b}}{(1-z)^{2a-1}} 
  \Big\{ -2 a (1-z)^{2a-1}\,{_2}F_1 \left ( a,b;c;z^2 \right )
\nonumber\\
&&\hspace{3cm}
+  (1-z)^{2a}\, \frac{\d}{\d z} {_2}F_1 \left (a,b;c;z^2 \right ) \Big\}
\,.
\end{eqnarray}
So we obtain the ``compact'' representation
\begin{equation}
Hl ( -1,-pq ; 0,2 q, 2q+1,p/2;z) 
= -(1-z)^{2(1-q)} \left ( \frac{1+z}{1-z} \right )^{\frac{p}{2}}\,F_{p,q}'(z)\,,
\label{hypergeometricFormula}
\end{equation}
where
\begin{equation}
F_{p,q}(z) :=  \frac{(1-z)^{\frac{p}{2} + 2q-1}}{\frac{p}{2} + 2q-1}\,{_2}F_1 \left ( {\frac{p}{4} +q - \frac{1}{2}} ,\frac{p}{4}; q+\frac{1}{2};z^2 \right )\,.
\label{Ffunction}
\end{equation}

\subsection{The associated Legendre function}
\label{legendreSubsection}

Formula (6) in \cite{Er}, \S 3.13, expresses the hypergeometric function in terms of the associated Legendre function:
\begin{eqnarray}
\Gamma \left ( q + \frac{1}{2} \right )\,Q_{q-1}^{{\frac{p}{4} - \frac{1}{2}}} \left ( \text{cosh}\,\eta \right ) 
= 2 \sqrt{\pi}\,\e^{\i \left ( {\frac{p}{4} - \frac{1}{2}} \right ) \pi}\,\Gamma \left ( {\frac{p}{4} + q + \frac{1}{2}} \right ) \nonumber\\
\times \left [ \frac{\e^{-\eta}}{\left ( 1-\e^{-\eta} \right )^2} \right ]^q\,
\left ( \frac{1+\e^{-\eta}}{1-\e^{-\eta}} \right )^{{\frac{p}{4} - \frac{1}{2}}}\,F_{p,q} \left ( \e^{-\eta} \right )\,.
\label{legendreFormula}
\end{eqnarray} 

Equation (8) in \cite{GR}, \S 7.621, gives a  useful integral representation of the Legendre function in terms of a Whittaker function:
\begin{eqnarray}
\nonumber
\int_0^\infty \e^{-\frac{x}{4} \left (z +1/z \right )} M_{{\frac{p}{4} - \frac{1}{2}},q-\frac{1}{2}} (x)\,\frac{\d x}{x} \\
= 2 \frac{\Gamma(2q)}{\Gamma(q+{\frac{p}{4} - \frac{1}{2}})} \e^{-\i {\frac{p}{4} - \frac{1}{2}} \pi} \left ( \frac{1-z}{1+z} \right )^{\frac{p}{4} - \frac{1}{2}} Q_{q-1}^{\frac{p}{4} - \frac{1}{2}} \left ( \frac{z+1/z}{2}\right )\,.
\end{eqnarray}
Then, by Equation (\ref{legendreFormula}), we obtain
\begin{eqnarray}
-\left [ \frac{z}{(1-z)^2} \right ]^q\,\frac{1-z}{1+z}\,F_{p,q}'(z) 
= \frac{1}{4 \sqrt{\pi} ({\frac{p}{4} - \frac{1}{2}}+q)} \frac{\Gamma
  (q+1/2)}{\Gamma(2q)} 
\nonumber\\
\times \int_0^\infty \e^{-\frac{x}{4} \left (z +1/z \right )} M_{{\frac{p}{4} - \frac{1}{2}},q-\frac{1}{2}} (x) \left [ \frac{q}{xz} - \left ( \frac{1-z}{2 z} \right )^2 \right ] \,\d x \,.
\label{intermediateEquation}
\end{eqnarray}
Now,
\begin{eqnarray}
\nonumber
M_{{\frac{p}{4} - \frac{1}{2}},q-\frac{1}{2}} (x) \left [ \frac{q}{xz} - \left ( \frac{1-z}{2 z} \right )^2 \right ] 
= M_{{\frac{p}{4} - \frac{1}{2}},q-\frac{1}{2}} (x) \left [ \left ( \frac{q}{x} + \frac{1}{2} \right ) \frac{1}{z} - \frac{1+z^2}{4 z^2} \right ] \\
= \frac{\frac{p}{4}+q-\frac{1}{2}}{2 q \sqrt{x}} \frac{1}{z} M_{\frac{p}{4},q}(x) +
\frac{1}{z} M_{\frac{p}{4} - \frac{1}{2},q-\frac{1}{2}}'(x) - \frac{1+z^2}{4 z^2} M_{\frac{p}{4} - \frac{1}{2},q-\frac{1}{2}}(x)
\end{eqnarray}
where we have used Formulae (13.4.28) and (13.4.32) from \cite{AS} to obtain the last equality. The proof
of Theorem \ref{whittakerTheorem} follows after reporting this in Equation (\ref{intermediateEquation})
and using integration by parts for the term involving the derivative of the Whittaker function. 

\vspace{0.5cm}

\noindent
{\it Remark.}
Formula (5) in \cite{Er}, \S 2.8, expresses $F_{p,0}$ in simple terms:
\begin{equation}
F_{p,0}(z) = \frac{1}{p-2} \left [ \left ( \frac{1-z}{1+z}\right )^{\frac{p}{2}-1} + 1 \right ]\,.
\end{equation}
It then follows from Equation (\ref{legendreFormula}) and from the recurrence relation for the associated Legendre function
(see Formula (18) in \cite{Er}, \S 3.8) that, for $q \in {\mathbb N}$, $F_{p,q}$ may be expressed in terms of elementary functions.
For example, 
we have
\begin{equation}
F_{p,1}(z) = \frac{1/8}{ \left ( \frac{p}{4} - \frac{1}{2} \right ) \left ({\frac{p}{4} + \frac{1}{2}} \right )} \frac{(1-z)^2}{z} \left [ 1 - \left ( \frac{1-z}{1+z} \right )^{{\frac{p}{2} - 1}} \right ]
\end{equation}
and
\begin{eqnarray}
\nonumber
F_{p,2}(z) = \frac{3/32}{\left ({\frac{p}{4} - \frac{3}{2}} \right )
  \left ( \frac{p}{4} - \frac{1}{2} \right ) \left ({\frac{p}{4} +
      \frac{1}{2}} \right ) \left ({\frac{p}{4} + \frac{3}{2}} \right
  )}
 \\
\hspace{-1.5cm}
\times \,\frac{(1-z)^4}{z^3} \left \{ \left ( \frac{1-z}{1+z} \right )^{{\frac{p}{2} - 1}} \left [ z^2+ \left ( \frac{p}{2} - 1 \right ) z +1\right ]- \left [
 z^2- \left ( \frac{p}{2} - 1 \right ) z +1 \right ] \right \}\,.
\end{eqnarray}
Formulae for $f_{p,1,1}$ and $f_{p,2,1}$ are easily deduced.

\subsection{The characteristic function}
\label{characteristicSubsection}

Define
\begin{equation}
\hat{f}_{p,q,k} (s) := \int_0^\infty z^{-s} f_{p,q,k}(z)\,\d z\,.
\end{equation}
Using Theorem \ref{whittakerTheorem} and changing the order of integration, we find 
\begin{equation}
\label{integralFormula}
\hspace{-2cm}
\hat{f}_{p,q,1} (s) = 2\,C\,\int_0^\infty M_{\frac{p}{4},q} (x) \,K_s \left ( \frac{x}{2} \right )\,\frac{\d x}{\sqrt{x}} 
= 2\,\sqrt{\pi}\,C\,\int_0^\infty M_{\frac{p}{4},q} (x) \,W_{0,s}(x)\,\frac{\d x}{x}\,.
\end{equation}
Next, Formula (7) in \cite{Er}, \S 6.9, says
\begin{equation*}
M_{\frac{p}{4},q}(x) = (-1)^{q+\frac{1}{2}} M_{-\frac{p}{4},q}(-x)\,.
\end{equation*}
Hence
\begin{equation}
\hat{f}_{p,q,1}(s) = 2 \,\sqrt{\pi} \,C\,(-1)^{q+\frac{1}{2}} \,\int_0^\infty M_{-\frac{p}{4},q}(-x) \,W_{0,s}(x)\,\frac{\d x}{x}\,.
\label{whittakerFormula}
\end{equation}
The integral on the right-hand side is of the same form as that in \cite{GR}, \S 7.625, Formula (1)--- albeit outside the range indicated since in our case
$\alpha = -1$ and $\beta = 1$. Nevertheless, direct numerical verification indicates that the formula does remain valid and so
\begin{eqnarray}
\nonumber
\hat{f}_{p,q,1}(s) = 2\,\sqrt{\pi} \,C\,\frac{\Gamma (q+1+s) \Gamma (q+1-s)}{\Gamma \left ( q+\frac{3}{2} \right )} \\
\times {_3}F_2 \left ( q+\frac{1}{2}-\frac{p}{4},q+1+s,q+1-s;\,2q+1,q+\frac{3}{2};\,1 \right )\,.
\end{eqnarray}
The normalisation condition
\begin{equation*}
\hat{f}_{p,q,1}(0) = 1
\end{equation*}
yields
\begin{equation}
C = \frac{\Gamma \left ( q + \frac{3}{2} \right )}{2 \sqrt{\pi}\,\Gamma (q+1)^2\,{_3}F_2 \left ( q+\frac{1}{2}-\frac{p}{4},q+1,q+1;\,2q+1,q+\frac{3}{2};\,1 \right )}
\label{normalisationConstant}
\end{equation}
and so we deduce the formula:
\begin{eqnarray}
\nonumber
\hat{f}_{p,q,1}(s) = \frac{\Gamma (q+1+s) \Gamma (q+1-s)}{\Gamma ( q+1)^2} \\
\times \,\frac{{_3}F_2 \left ( q+\frac{1}{2}-\frac{p}{4},q+1+s,q+1-s;\,2q+1,q+\frac{3}{2};\,1 \right )}{{_3}F_2 \left ( q+\frac{1}{2}-\frac{p}{4},q+1,q+1;\,2q+1,q+\frac{3}{2};\,1 \right )}
\end{eqnarray}
for $-1 \le \text{Re} \,s \le 1$. Then,
using Formula (\ref{riccatiFormula}) for the Lyapunov exponent, together with the identity (\ref{negativeEnergyIdentity}) and the fact that $N(E)=0$ for $E<0$, 
we eventually find
\begin{equation}
\label{eq:OmegaGeneralCase}
\Omega(-k^2) = k \frac{q+1}{q} 
\frac{{_3}F_2 \left ( q+\frac{1}{2}-\frac{p}{4k},q+2,q;\,2q+1,q+\frac{3}{2};\,1 \right )}{{_3}F_2 \left ( q+\frac{1}{2}-\frac{p}{4k},q+1,q+1;\,2q+1,q+\frac{3}{2};\,1 \right )}\,.
\end{equation}
The analytic continuation from negative to positive energy consists of replacing $k$ by $-\i k$. This completes the proof of Corollary \ref{whittakerCorollary}.


\section*{References}

\end{document}